\documentclass[a4paper]{article}
\usepackage{cite}
\usepackage{algorithm}
\usepackage{amsmath,amssymb,amsfonts}
\usepackage{algorithmic}
\usepackage[dvipdfmx]{graphicx}
\usepackage{color}
\usepackage{textcomp}
\usepackage{url}
\newtheorem{Theorem}{Theorem}[section]

\newtheorem{Remark}{Remark}[section]

\newtheorem{Definition}{Definition}[section]

\newtheorem{proposition}{Proposition}
\newtheorem{proof}{Proof}

\DeclareMathOperator*{\minimize}{minimize}

\definecolor{accessblue}{cmyk}{1,  0.3,  0,  0.2}
\definecolor{greycolor}{cmyk}{0, 0, 0, .8}

\begin{document}




\title{LMI Optimization Based Multirate Steady-State Kalman Filter Design}
\author{ Hiroshi Okajima}

\begin{abstract}
This paper presents an LMI-based design framework for multirate steady-state Kalman filters in systems with sensors operating at different sampling rates. The multirate system is formulated as a periodic time-varying system, where the Kalman gains converge to periodic steady-state values that repeat every frame period. Cyclic reformulation transforms this into a time-invariant problem; however, the resulting measurement noise covariance becomes semidefinite rather than positive definite, preventing direct application of standard Riccati equation methods. I address this through a dual LQR formulation with LMI optimization that naturally handles semidefinite covariances. The framework enables multi-objective design, supporting pole placement for guaranteed convergence rates and $l_2$-induced norm constraints for balancing average and worst-case performance. Numerical validation using an automotive navigation system with GPS and wheel speed sensors, including Monte Carlo simulation with 500 independent noise realizations, demonstrates that the proposed filter achieves a position RMSE well below the GPS noise level through effective multirate sensor fusion, and that the LMI solution provides valid upper bounds on the estimation error covariance.
\end{abstract}



\maketitle

\section{Introduction}\label{sec:introduction}

In many modern control and signal processing applications, state estimation must be performed using measurements from multiple sensors operating at different sampling rates. This multirate sensing environment arises naturally in systems where sensor characteristics, communication bandwidth, or computational constraints lead to varying measurement availability. Examples include industrial process control, robotics (visual and inertial sensor fusion), aerospace systems, autonomous vehicles, and networked control systems. The challenge in such scenarios is to optimally fuse measurements arriving at different rates to achieve accurate state estimation.

The Kalman filter has been extensively used for state estimation in linear systems \cite{kalman1960}. However, the classical Kalman filter formulation assumes that all measurements are available at a uniform sampling rate. When sensors operate at different rates, the system exhibits periodic time-varying dynamics in the measurement equations, even if the underlying plant dynamics are time-invariant. Naive approaches such as zero-order hold or measurement interpolation can lead to suboptimal performance and may not properly account for the uncertainty structure of the multirate measurements.

Multirate state estimation is essential across diverse engineering domains. In autonomous vehicles, GPS (1--10 Hz) and inertial/wheel speed sensors (50--100 Hz) must be fused for navigation \cite{singh2023extended,almeida2019autonomous,armesto2007egomotion}. Similar challenges arise in chemical process control where laboratory analysis and online sensors operate at different rates \cite{hellmann2024tutorial}, in power systems combining PMUs with SCADA, and in networked control systems with bandwidth-limited or event-triggered sampling \cite{fadali2002timely}. Mobile robotics applications, including SLAM with laser and encoder fusion \cite{armesto2004slam}, further demonstrate the widespread need for effective multirate estimation techniques.

Several approaches address multirate Kalman filtering. The lifting technique \cite{zhang2002multiresolution, li2008kalman, li2005kalman, armesto2006lqg} transforms multirate systems into equivalent time-invariant systems at the frame period by augmenting input and output vectors. While preserving state dimension, lifting requires matrix products such as $A^N$ and operates in open-loop between frame periods, limiting inter-sample accuracy. In contrast, the cyclic reformulation \cite{bittanti2000invariant} employed in this paper maintains the original state dimension within each block and explicitly represents periodicity through cyclic shift of state variables, enabling LMI-based optimization that naturally handles semidefinite measurement noise covariances. For systems with measurement delays, pseudo-measurement approaches \cite{nguyen2017dualrate, wang2014multirate} reconstruct current-time observations from delayed measurements and predicted states. Asynchronous filtering \cite{federated2022multirate, meng2011comparison} processes measurements upon arrival but needs careful update handling.

The theoretical foundation of periodic systems and periodic Riccati equations has been extensively studied. Bittanti, Colaneri, and their collaborators established fundamental results on discrete-time periodic Riccati equations (DPRE), including existence, uniqueness, and convergence properties of periodic solutions \cite{bittanti1988dpre, bittanti1986analysis}. Varga \cite{varga2008periodic} developed numerically reliable algorithms for solving both continuous-time and discrete-time periodic Riccati equations using periodic Schur decomposition. De Souza \cite{desouza1991periodic} addressed the periodic Riccati difference equation specifically for optimal filtering of linear periodic discrete-time systems. For systems with singular measurement noise covariance, Nikoukhah et al. \cite{nikoukhah1992descriptor} studied Kalman filtering and generalized Riccati equations in the descriptor system framework.

This paper proposes a systematic design method for multirate Kalman filters based on cyclic reformulation, building upon the periodic Kalman filtering approach developed by Fujimoto et al. \cite{fujimoto2016periodic,uehara2017multirate}. Their method transforms periodic systems into time-invariant dilated systems and enables offline computation of optimal Kalman gains through a discrete-time algebraic Riccati equation (DARE). However, the standard DARE formulation assumes positive definite measurement noise covariance, which does not hold in realistic multirate scenarios where sensors operate at truly different rates.


The main contributions of this paper are as follows. First, I identify that the cyclic measurement noise covariance $\check{R}$ is positive semidefinite but not positive definite in realistic multirate systems, and present an LMI-based approach via dual LQR formulation that handles this naturally. Second, I develop LMI-based extensions for multi-objective design, including pole placement constraints for guaranteed convergence rates and Kalman filter design with $l_2$-induced norm constraints \cite{okajima2019multirate,okajima2023multirate} for balancing average and worst-case performance. Third, I demonstrate the effectiveness of the proposed methods through automotive state estimation simulations with GPS and wheel speed sensors operating at different rates. The LMI framework also provides additional benefits: stability guarantees through the established connection between cyclic system stability and the monodromy matrix, and straightforward extension to weighted optimal estimation by modifying the objective function.

The remainder of this paper is organized as follows. Section~\ref{sec:preliminaries} reviews the standard Kalman filter formulation and periodic time-varying systems. Section~\ref{sec:design} formulates the multirate Kalman filtering problem using cyclic reformulation and identifies that the resulting measurement noise covariance is semidefinite, precluding standard DARE methods. Section~\ref{lmibased} develops the LMI-based filter design and weighted estimation. Section~\ref{sec:multi_obj} extends the framework with pole placement and $l_2$-induced norm constraints. Section~\ref{sec:numerical} demonstrates the proposed methods through numerical simulations. Finally, Section~\ref{sec:conclusion} concludes the paper.

\section{Preliminaries}\label{sec:preliminaries}

\subsection{Standard Kalman Filter}

Consider a discrete-time linear time-invariant system:
\begin{align}
x(k+1) &= Ax(k) + Bu(k) + w(k) \label{eq:system1}\\
y(k) &= Cx(k) + v(k) \label{eq:measurement1}
\end{align}
where $x(k) \in \mathbb{R}^n$ is the state vector, $u(k) \in \mathbb{R}^p$ is the control input, $y(k) \in \mathbb{R}^m$ is the measurement vector, $w(k) \in \mathbb{R}^n$ is the process noise, and $v(k) \in \mathbb{R}^m$ is the measurement noise. The process and measurement noises are assumed to be zero-mean Gaussian white noise with covariance matrices $Q$ and $R$, respectively, where $\delta_{kj}$ is the Kronecker delta:
\begin{equation}
\mathbb{E}[w(k)w(j)^T] = Q\delta_{kj}, \quad \mathbb{E}[v(k)v(j)^T] = R\delta_{kj}
\end{equation}

The discrete-time Kalman filter consists of prediction and update steps:

\textbf{Prediction:}
\begin{align}
\hat{x}(k+1|k) &= A\hat{x}(k|k) + Bu(k) \label{eq:kf_pred1}\\
P(k+1|k) &= AP(k|k)A^T + Q \label{eq:kf_pred2}
\end{align}

\textbf{Update:}
\begin{align}
K(k) &= P(k|k-1)C^T(CP(k|k-1)C^T + R)^{-1} \label{eq:kf_gain}\\
\hat{x}(k|k) &= \hat{x}(k|k-1) + K(k)(y(k) - C\hat{x}(k|k-1)) \label{eq:kf_update1}\\
P(k|k) &= (I - K(k)C)P(k|k-1) \label{eq:kf_update2}
\end{align}
where $K(k)$ is the Kalman gain matrix, $\hat{x}(k|k)$ is the estimated state given measurements up to time $k$, and $P(k|k)$ is the estimation error covariance.

For time-invariant systems, the Kalman gain converges to a steady-state value $K_{ss}$ as $k \to \infty$, which can be computed by solving the discrete-time algebraic Riccati equation (DARE):
\begin{equation}
P = APA^T - APC^T(CPC^T + R)^{-1}CPA^T + Q \label{eq:dare}
\end{equation}
The steady-state Kalman gain is then:
\begin{equation}
K_{ss} = PC^T(CPC^T + R)^{-1} \label{eq:kf_gain_ss}
\end{equation}

In this paper, I adopt the predictor form \cite{anderson1979optimal}, which is standard in observer-based control design. The filter-form gain $K_{ss}$ is converted to the predictor-form gain $L_{ss} = AK_{ss}$ by premultiplying by $A$, so that the measurement correction is incorporated into the one-step-ahead state prediction. The corresponding estimator is:
\begin{equation}
\hat{x}(k+1) = A\hat{x}(k) + Bu(k) + L_{ss}(y(k) - C\hat{x}(k)) \label{eq:predictor_form}
\end{equation}
Defining the estimation error $e(k) = x(k) - \hat{x}(k)$, the error dynamics becomes:
\begin{equation}
e(k+1) = (A - L_{ss}C)e(k) + w(k) - L_{ss}v(k) \label{eq:error_dynamics}
\end{equation}
This formulation facilitates LMI-based stability analysis through the closed-loop matrix $A - L_{ss}C$.
\subsection{Periodic Time-Varying Systems}

A discrete-time periodic system with period $N$ is characterized by:
\begin{align}
x(k+1) &= A_k x(k) + B_k u(k) + w(k) \label{eq:periodic_sys}\\
y(k) &= C_k x(k) + v(k) \label{eq:periodic_meas}
\end{align}
where $A_{k+N} = A_k$, $B_{k+N} = B_k$, and $C_{k+N} = C_k$ for all $k \geq 0$. The matrices $A_k$, $B_k$, and $C_k$ vary periodically with period $N$.

For periodic systems, the Kalman filter gains in the predictor form also become periodic. As $k \to \infty$, the gains converge to \textbf{periodic steady-state} values $L_k$, $k = 0, 1, \ldots, N-1$, satisfying:
\begin{equation}
L_{k+N} = L_k, \quad k = 0, 1, \ldots, N-1
\end{equation}
This is analogous to the steady-state gain $L_{ss}$ in time-invariant systems, but the gains now form a periodic sequence rather than a constant. Computing these periodic gains requires solving coupled Riccati equations, which is computationally intensive for large $N$. Sections~\ref{sec:design} and \ref{lmibased} show how cyclic reformulation and LMI optimization transform this into an efficiently solvable convex problem.

\section{Multirate Kalman Filter Expression}\label{sec:design}

\subsection{Problem Formulation}

Consider a linear time-invariant discrete-time system:
\begin{equation}
x(k+1) = Ax(k) + Bu(k) + Q^{1/2} d_w(k) \label{eq:system_general}
\end{equation}
where $x(k) \in \mathbb{R}^n$ is the state vector, $u(k) \in \mathbb{R}^p$ is the control input, and $d_w(k) \in \mathbb{R}^n$ is a normalized disturbance with zero mean and unit covariance. The matrix $Q^{1/2}$ is obtained from the Cholesky factorization $Q = Q^{1/2}(Q^{1/2})^T$, where $Q \succ 0$ is the process noise covariance.

Suppose the system is equipped with multiple sensors operating at different sampling rates. Let $C \in \mathbb{R}^{q \times n}$ denote the full measurement matrix, where $q$ is the total number of scalar measurements. Since not all sensors provide measurements at every sampling instant, I introduce a periodic selection matrix $S_k \in \mathbb{R}^{q \times q}$ \cite{okajima2023multirate,okajima2024design} to represent measurement availability. The measurement equation for the multirate system is:
\begin{equation}
y(k) = S_k C x(k) + S_k R^{1/2} d_v(k) \label{eq:multirate_meas_general}
\end{equation}
where $d_v(k) \in \mathbb{R}^q$ is a normalized disturbance with unit covariance, and $R^{1/2}$ is obtained from $R = R^{1/2}(R^{1/2})^T$ with $R \succ 0$ being the measurement noise covariance. The matrix $S_k$ is a diagonal matrix with entries of 0 or 1:
\begin{equation}
S_k = \mathrm{diag}(s_{k,1}, s_{k,2}, \ldots, s_{k,q})
\end{equation}
where $s_{k,i} = 1$ if the $i$-th measurement is available at time $k$, and $s_{k,i} = 0$ otherwise.

If the sampling rates of different sensors have rational ratios, there exists a period $N$ such that:
\begin{equation}
S_{k+N} = S_k, \quad \forall k \geq 0
\end{equation}

The multirate Kalman filtering problem is to design optimal periodic Kalman gains $L_k$, $k = 0, 1, \ldots, N-1$, for the predictor-form estimator:
\begin{equation}
\hat{x}(k+1) = A\hat{x}(k) + Bu(k) + L_k(y(k) - S_kC\hat{x}(k)) \label{eq:multirate_estimator_general}
\end{equation}
where $\hat{x}(k)$ denotes $\hat{x}(k|k-1)$ as introduced in Section~\ref{sec:preliminaries}. The objective is to minimize the estimation error $e(k) = x(k) - \hat{x}(k)$ in the mean square sense.

\subsection{Cyclic Reformulation}

To transform the periodic system (\ref{eq:system_general})--(\ref{eq:multirate_meas_general}) into a time-invariant form, I apply cyclic reformulation \cite{bittanti2000invariant}, following its extension to multirate systems in \cite{okajima2019multirate,okajima2024design}. For a signal $\xi(k)$, define its cycled representation $\check{\xi}(k)$ by placing the signal in the $(k \bmod N + 1)$-th block position with zeros elsewhere. Specifically, for $k \bmod N = j$ $(j = 0, 1, \ldots, N-1)$:
\begin{equation}
\check{\xi}(k) = \begin{bmatrix} O \\ \vdots \\ O \\ \xi(k) \\ O \\ \vdots \\ O \end{bmatrix} \leftarrow (j+1)\text{-th block}
\end{equation}

For example, the cycled input $\check{u}(k) \in \mathbb{R}^{Np}$ is defined as:
\begin{align}
\check{u}(0) &= \begin{bmatrix} u(0) \\ O \\ \vdots \\ O \end{bmatrix}, \quad 
\check{u}(1) = \begin{bmatrix} O \\ u(1) \\ \vdots \\ O \end{bmatrix}, \quad \cdots, \notag \\
&\check{u}(N-1) = \begin{bmatrix} O \\ \vdots \\ O \\ u(N-1) \end{bmatrix}, \quad \cdots
\end{align}
This cycled representation is applied to all signals: $\check{x}(k) \in \mathbb{R}^{Nn}$ for the state, $\check{u}(k) \in \mathbb{R}^{Np}$ for the input, $\check{y}(k) \in \mathbb{R}^{Nq}$ for the output, $\check{d}_w(k) \in \mathbb{R}^{Nn}$ for the normalized process disturbance, and $\check{d}_v(k) \in \mathbb{R}^{Nq}$ for the normalized measurement disturbance. The cycled state dynamics is given by:
\begin{equation}
\check{x}(k+1) = \check{A}\check{x}(k) + \check{B}\check{u}(k) + \check{Q}^{1/2} \check{d}_w(k) \label{eq:cyclic_dynamics}
\end{equation}
where the cyclic system matrix $\check{A} \in \mathbb{R}^{Nn \times Nn}$ has the structure:
\begin{equation}
\check{A} = \begin{bmatrix}
0 & 0 & \cdots & 0 & A\\
A & 0 & \cdots & 0 & 0\\
0 & A & \cdots & 0 & 0\\
\vdots & \vdots & \ddots & \vdots & \vdots\\
0 & 0 & \cdots & A & 0
\end{bmatrix} \label{eq:cyclic_A}
\end{equation}

Similarly, $\check{B}$ and $\check{Q}^{1/2}$ follow the same cyclic pattern:
\begin{eqnarray}
\check{B} &=& \begin{bmatrix}
0 & 0 & \cdots & 0 & B\\
B & 0 & \cdots & 0 & 0\\
0 & B & \cdots & 0 & 0\\
\vdots & \vdots & \ddots & \vdots & \vdots\\
0 & 0 & \cdots & B & 0
\end{bmatrix}, \\ 
\check{Q}^{1/2} &=& \begin{bmatrix}
0 & 0 & \cdots & 0 & Q^{1/2}\\
Q^{1/2} & 0 & \cdots & 0 & 0\\
0 & Q^{1/2} & \cdots & 0 & 0\\
\vdots & & \ddots & & \vdots\\
0 & 0 & \cdots & Q^{1/2} & 0
\end{bmatrix} \label{eq:cyclic_Q}
\end{eqnarray}

The cycled measurement equation is:
\begin{equation}
\check{y}(k) = \check{C}\check{x}(k) + \check{R}^{1/2} \check{d}_v(k) \label{eq:cyclic_measurement}
\end{equation}
where $\check{C} \in \mathbb{R}^{Nq \times Nn}$ and $\check{R}^{1/2} \in \mathbb{R}^{Nq \times Nq}$ are block-diagonal matrices:
\begin{align}
\check{C} &= \mathrm{diag}(S_0C, S_1C, \ldots, S_{N-1}C) \label{eq:cyclic_C} \\
\check{R}^{1/2} &= \mathrm{diag}(S_0 R^{1/2}, S_1 R^{1/2}, \ldots, S_{N-1} R^{1/2}) \label{eq:cyclic_R}
\end{align}

\begin{proposition}[\cite{bittanti2000invariant}]
The cyclic reformulation $(\check{A}, \check{B}, \check{C},$ $\check{Q}^{1/2}, \check{R}^{1/2})$ is equivalent to the original periodic system $(A, B, S_k C)$ with noise covariances $(Q, S_k R S_k^T)$ in the sense that if $\check{x}(0) = [x(0)^T, 0, \ldots, 0]^T$, then the non-zero block of $\check{x}(k)$ always corresponds to $x(k)$ for all $k \geq 0$.
\end{proposition}

The cyclic noise covariances are given by:
\begin{align}
\check{Q} &= \check{Q}^{1/2} (\check{Q}^{1/2})^T = \mathrm{diag}(Q, Q, \ldots, Q) \label{eq:Q_cyclic_cov}\\
\check{R} &= \check{R}^{1/2} (\check{R}^{1/2})^T \nonumber \\ &= \mathrm{diag}(S_0 R S_0^T, S_1 R S_1^T, \ldots, S_{N-1} R S_{N-1}^T) \label{eq:R_cyclic_cov}
\end{align}

\begin{Remark}[Critical Observation]
The measurement noise covariance $\check{R}$ is only positive semidefinite ($\check{R} \succeq 0$), not positive definite. This is because at time instants when certain sensors do not provide measurements (i.e., $S_k$ has zero rows), the corresponding blocks $S_k R S_k^T$ are singular or have reduced rank.
\end{Remark}

For example, in the automotive navigation system:
\begin{itemize}
\item When GPS is available ($k \bmod N = 0$): $S_0 = I_q$ $\Rightarrow$ $S_0 R S_0^T = R \succ 0$ (full rank)
\item When only wheel speed is available ($k \bmod N \neq 0$): $S_k = \mathrm{diag}(0, 1)$ $\Rightarrow$ $S_k R S_k^T$ has reduced rank
\end{itemize}

Therefore, $\check{R} \succeq 0$ but $\check{R} \not\succ 0$. This semidefinite structure prevents the direct application of standard DARE methods and necessitates an LMI-based approach.

\subsection{Observability Analysis}

The observability of the cyclic system is crucial for the existence of a stable Kalman filter. The observability matrix for the $Nn$-dimensional cyclic system $(\check{A}, \check{C})$ is:
\begin{equation}
\mathcal{O} = \begin{bmatrix}
\check{C}\\
\check{C}\check{A}\\
\check{C}\check{A}^2\\
\vdots\\
\check{C}\check{A}^{Nn-1}
\end{bmatrix}
\end{equation}

\begin{proposition}[\cite{bittanti2000invariant,okajima2024design}]\label{prop:observability}
The cyclic system $(\check{A}, \check{C})$ is observable if and only if $\mathrm{rank}(\mathcal{O}) = Nn$. Since the system matrices $A$, $S_k$, and $C$ are assumed known, this condition can be verified directly by computing the rank of $\mathcal{O}$.
\end{proposition}

This is a collective observability condition over one complete period of $N$ time steps: the combination of all measurements collected over one period must provide sufficient information to reconstruct the full state. Individual time steps need not be independently observable.

\begin{Remark}[Sufficient Conditions and Detectability]\label{rem:obs_sufficient}
A sufficient condition for the observability of $(\check{A}, \check{C})$ is that there exists some $j \in \{0, \ldots, N-1\}$ such that the pair $(S_j C, A^N)$ is observable and $A$ is nonsingular \cite{okajima2025jrm}. In this case, the nonsingularity of $A$ propagates the observability at a single measurement time instant to all block positions in the cyclic observability matrix, establishing $\mathrm{rank}(\mathcal{O}) = Nn$. For the stability of the Kalman filter, the weaker condition of detectability of $(\check{A}, \check{C})$ is sufficient: all unobservable modes of the cyclic system must be Schur stable. In particular, when $A$ is Schur stable, detectability is automatically satisfied regardless of the measurement pattern.
\end{Remark}

This result ensures that the cyclic reformulation preserves the observability properties of the original multirate system, which is essential for the well-posedness of the Kalman filtering problem.

\subsection{Challenges with Standard DARE Approach}

The key advantage of cyclic reformulation is that it transforms the periodic Kalman filtering problem into a time-invariant problem. Following the approach of Fujimoto et al. \cite{fujimoto2016periodic}, one might attempt to compute the optimal cyclic Kalman gain $\check{L}$ by solving the standard DARE:
\begin{equation}
\check{P} = \check{A}\check{P}\check{A}^T - \check{A}\check{P}\check{C}^T(\check{C}\check{P}\check{C}^T + \check{R})^{-1}\check{C}\check{P}\check{A}^T + \check{Q} \label{eq:cyclic_dare}
\end{equation}
with the Kalman gain:
\begin{equation}
\check{L} = \check{A}\check{P}\check{C}^T(\check{C}\check{P}\check{C}^T + \check{R})^{-1} \label{eq:cyclic_gain}
\end{equation}

\begin{Remark}
The standard DARE formulation (\ref{eq:cyclic_dare}) requires $\check{R} \succ 0$ to ensure that the inverse $(\check{C}\check{P}\check{C}^T + \check{R})^{-1}$ exists and is well-conditioned. In realistic multirate systems where sensors operate at different rates, the cyclic covariance $\check{R}$ is only positive semidefinite. When $\check{R}$ is singular, standard DARE solvers may fail to converge or produce numerically unstable solutions. Note that Fujimoto et al. \cite{fujimoto2016periodic} considered systems with $\check{R} \succ 0$; the semidefinite case necessitates the LMI-based approach developed in Section~\ref{lmibased}.
\end{Remark}

\subsection{Stability Analysis}

The closed-loop estimation error dynamics for the cyclic system is given by:
\begin{equation}
\check{e}(k+1) = (\check{A} - \check{L}\check{C})\check{e}(k) + \check{Q}^{1/2}\check{d}_w(k) - \check{L}\check{R}^{1/2}\check{d}_v(k)
\end{equation}

The stability of the multirate Kalman filter is determined by the spectral radius of the closed-loop matrix $\check{A} - \check{L}\check{C}$. A fundamental connection exists between the cyclic system stability and the monodromy matrix of the original periodic system.

\begin{Definition}[Monodromy Matrix]
For the periodic estimation error system with gains $L_k$, $k = 0, 1, \ldots, N-1$, the monodromy matrix is defined as:
\begin{equation}
\Phi_N = \prod_{k=0}^{N-1} (A - L_{N-1-k} S_{N-1-k} C) \label{eq:monodromy}
\end{equation}
\end{Definition}

\begin{Theorem}[Equivalence of Stability Conditions \cite{bittanti2000invariant, bittanti1986analysis}]\label{stabtheo}
The following statements are equivalent:
\begin{enumerate}
\item The cyclic closed-loop matrix $\check{A} - \check{L}\check{C}$ is Schur stable.
\item The monodromy matrix $\Phi_N$ is Schur stable.
\item The estimation error covariance remains bounded for all time.
\end{enumerate}
\end{Theorem}

\begin{proof}
The equivalence between (i) and (ii) follows from the spectral mapping property of cyclic systems. Specifically, if $\lambda$ is an eigenvalue of $\Phi_N$, then $\lambda^{1/N} e^{j2\pi m/N}$ for $m = 0, 1, \ldots, N-1$ are eigenvalues of $\check{A} - \check{L}\check{C}$. Consequently:
\begin{equation}
\rho(\check{A} - \check{L}\check{C}) = \rho(\Phi_N)^{1/N}
\end{equation}
where $\rho(\cdot)$ denotes the spectral radius. Thus $\rho(\Phi_N) < 1$ if and only if $\rho(\check{A} - \check{L}\check{C}) < 1$. The equivalence between (ii) and (iii) follows from periodic Lyapunov theory: the periodic Lyapunov equation for the error covariance has a bounded positive definite solution if and only if $\Phi_N$ is Schur stable.
\end{proof}

Theorem~\ref{stabtheo} also guarantees the existence of periodic steady-state Kalman gains \cite{desouza1991periodic}: when $\Phi_N$ is Schur stable, the periodic Riccati difference equation has a unique periodic positive definite solution, and the time-varying Kalman gains converge to periodic steady-state values $L_k$ as $k \to \infty$. This justifies the offline computation of periodic gains.

The complete design procedure, including the stability guarantee of the resulting filter, is presented in Section~\ref{lmibased}.

\section{LMI-based Filter Design for Multirate Systems} \label{lmibased}
\subsection{LMI-based Filter Design}

As established in Section \ref{sec:design}, multirate systems with intermittent sensor measurements lead to a semidefinite measurement noise covariance $\check{R} \succeq 0$ (not positive definite). The LMI-based approach offers several advantages including natural handling of semidefinite covariances, easier incorporation of robustness constraints, and guaranteed numerical stability.

\subsubsection{Error Dynamics}

The estimation error $\check{e}(k) = \check{x}(k) - \check{\hat{x}}(k)$ satisfies:
\begin{equation}
\check{e}(k+1) = (\check{A} - \check{L}\check{C})\check{e}(k) + \check{Q}^{1/2}\check{d}_w(k) - \check{L}\check{R}^{1/2}\check{d}_v(k) \label{eq:error_dyn_cyclic}
\end{equation}
For any stabilizing gain $\check{L}$, the steady-state error covariance $\check{P}_e = \lim_{k\to\infty}\mathbb{E}[\check{e}(k)\check{e}(k)^T]$ satisfies the discrete Lyapunov equation (assuming $\mathbb{E}[\check{d}_w \check{d}_v^T] = 0$):
\begin{equation}
\check{P}_e = (\check{A} - \check{L}\check{C})\check{P}_e(\check{A} - \check{L}\check{C})^T + \check{Q} + \check{L}\check{R}\check{L}^T \label{eq:lyap_cov}
\end{equation}
This equation holds for an arbitrary stabilizing $\check{L}$; when the optimal gain (\ref{eq:cyclic_gain}) is substituted, $\check{P}_e$ coincides with the DARE solution $\check{P}$ in (\ref{eq:cyclic_dare}). Note that, unlike the DARE (\ref{eq:cyclic_dare}), equation (\ref{eq:lyap_cov}) does not involve the inverse of $\check{C}\check{P}\check{C}^T + \check{R}$ and is therefore well-defined even when $\check{R}$ is only positive semidefinite.

By the duality between Kalman filtering and LQR \cite{anderson1979optimal}, the optimal gain $\check{L}$ for (\ref{eq:lyap_cov}) equals $\check{F}^T$, where $\check{F}$ is the LQR gain for the dual system $\check{\zeta}(k+1) = \check{A}^T \check{\zeta}(k) + \check{C}^T \check{u}(k)$ with state weight $\check{Q}$ and input weight $\check{R}$. This duality enables the application of standard LQR-LMI methods \cite{boyd1994lmi}, as developed below.

\subsubsection{LMI Formulation}

Introduce the upper bound variable $\check{X} \succ 0$ satisfying $\check{P}_e \preceq \check{X}^{-1}$. Replacing $\check{P}_e$ with $\check{X}^{-1}$ on both sides of (\ref{eq:lyap_cov}) gives the sufficient condition for maintaining this bound:
\begin{multline}
\check{X}^{-1} \succeq (\check{A} - \check{L}\check{C})\check{X}^{-1}(\check{A} - \check{L}\check{C})^T \\
+ \check{Q}^{1/2}(\check{Q}^{1/2})^T + \check{L}\check{R}^{1/2}(\check{R}^{1/2})^T\check{L}^T \label{eq:lyap_ineq}
\end{multline}
Define the change of variables $\check{Y} = -\check{X}\check{L}$. Pre- and post-multiplying (\ref{eq:lyap_ineq}) by $\check{X}$ and using $\check{X}(\check{A} - \check{L}\check{C}) = \check{X}\check{A} + \check{Y}\check{C}$:
\begin{multline}
\check{X} \succeq (\check{X}\check{A} + \check{Y}\check{C})\check{X}^{-1}(\check{X}\check{A} + \check{Y}\check{C})^T \\
+ \check{X}\check{Q}^{1/2}(\check{Q}^{1/2})^T\check{X} + \check{Y}\check{R}^{1/2}(\check{R}^{1/2})^T\check{Y}^T \label{eq:congruence}
\end{multline}
The right-hand side of (\ref{eq:congruence}) consists of three terms, each of the form $\Phi_i M_i^{-1} \Phi_i^T$ with $M_1 = \check{X}$, $M_2 = I_{Nn}$, $M_3 = I_{Nq}$. Applying the Schur complement to each term simultaneously, the matrix inequality (\ref{eq:congruence}) is equivalent to the following block matrix inequality.

\textbf{LMI \#1: Stability and Performance}
\begin{equation}
\begin{bmatrix}
\check{X} & \check{X}\check{A} + \check{Y}\check{C} & \check{X} \check{Q}^{1/2} & \check{Y} \check{R}^{1/2}\\
(\check{X}\check{A} + \check{Y}\check{C})^T & \check{X} & 0 & 0\\  
(\check{Q}^{1/2})^T \check{X} & 0 & I_{Nn} & 0\\
  (\check{R}^{1/2})^T \check{Y}^T & 0 & 0 & I_{Nq}
\end{bmatrix} \succeq 0 \label{eq:lmi_main}
\end{equation}

In this formulation, $\check{L}$ is absorbed into the decision variable $\check{Y} = -\check{X}\check{L}$, so that (\ref{eq:lmi_main}) is linear in $(\check{X}, \check{Y})$. The gain $\check{L}$ is treated as a free design variable, and minimizing an upper bound on trace$(\check{P}_e)$ (see LMI \#3 below) replaces the direct minimization of trace$(\check{P}_e)$ via the DARE. Since the LMI (\ref{eq:lmi_main}) never requires the inversion of $\check{R}$, it naturally accommodates semidefinite $\check{R}$.

Each block of (\ref{eq:lmi_main}) corresponds to a term in the error dynamics (\ref{eq:error_dyn_cyclic}): the (1,2) block $\check{X}\check{A} + \check{Y}\check{C}$ encodes the closed-loop dynamics $\check{A} - \check{L}\check{C}$; the (1,3) block $\check{X}\check{Q}^{1/2}$ corresponds to the process noise; and the (1,4) block $\check{Y}\check{R}^{1/2}$ corresponds to the measurement noise. The diagonal blocks $I_{Nn}$ and $I_{Nq}$ at positions (3,3) and (4,4) arise from the unit-covariance normalization of $\check{d}_w$ and $\check{d}_v$.


\textbf{LMI \#2: Positive Definiteness}
\begin{equation}
\check{X} \succeq \epsilon I_{Nn}, \quad \epsilon > 0 \label{eq:lmi_stability}
\end{equation}

\textbf{LMI \#3: Covariance Upper Bound}

To minimize trace$(\check{P}_e)$ where $\check{P}_e$ is the steady-state error covariance satisfying $\check{P}_e \preceq \check{X}^{-1}$, I introduce an auxiliary variable $\check{W} \in \mathbb{R}^{Nn \times Nn}$ and impose:
\begin{equation}
\begin{bmatrix}
\check{W} & I_{Nn}\\
I_{Nn} & \check{X}
\end{bmatrix} \succeq 0 \label{eq:lmi_bound}
\end{equation}
This constraint ensures $\check{W} \succeq \check{X}^{-1} \succeq \check{P}_e$. Therefore, minimizing trace$(\check{W})$ provides an upper bound minimization for trace$(\check{P}_e)$.

\subsubsection{Optimization Problem}

The complete LMI-based filter design problem is:
\begin{equation}
\begin{aligned}
\minimize_{\check{X}, \check{Y}, \check{W}} \quad & \text{trace}(\check{W}) \\
\text{subject to} \quad & \text{LMI (\ref{eq:lmi_main}), (\ref{eq:lmi_stability}), (\ref{eq:lmi_bound})}
\end{aligned}
\label{eq:lmi_optimization}
\end{equation}
This is a standard semidefinite programming (SDP) problem that can be efficiently solved using interior-point methods.

\subsubsection{Kalman Gain Extraction}

After solving the LMI optimization, the Kalman gain is recovered from the decision variables:
\begin{equation}
\check{L} = -\check{X}^{-1} \check{Y} \label{eq:L_recovery}
\end{equation}
The cyclic gain $\check{L} \in \mathbb{R}^{Nn \times Nq}$ must be decomposed into periodic gains $L_k \in \mathbb{R}^{n \times q}$ for $k = 0, 1, \ldots, N-1$ to be applied in the time-varying estimator (\ref{eq:multirate_estimator_general}). Due to the cyclic structure of $\check{A}$ in (\ref{eq:cyclic_A}), let $\check{L}_{(i),(j)}$ denote the $n \times q$ block at block-row $i$ and block-column $j$ of $\check{L}$. Then:
\begin{align}
L_0 &= \check{L}_{(2),(1)}, \quad L_k = \check{L}_{(k+2),(k+1)} \text{ for } k = 1, \ldots, N-2, \nonumber\\
L_{N-1} &= \check{L}_{(1),(N)}
\end{align}
This indexing reflects the cyclic shift structure where the measurement at time $k$ (located in block $k+1$ of $\check{y}$) affects the state estimate at time $k+1$ (located in block $(k+2) \bmod N$ of $\check{x}$).

The error covariance upper bound satisfies:
\begin{equation}
\check{P}_e \preceq \check{X}^{-1} \preceq \check{W}
\end{equation}

\subsubsection{Stability Guarantee}

\begin{Theorem}[Stability Guarantee via LMI Feasibility]\label{thm:lmi_stability}
If the LMI optimization (\ref{eq:lmi_optimization}) yields a feasible solution with $\check{X} \succ 0$, then the resulting filter gain $\check{L} = -\check{X}^{-1}\check{Y}$ guarantees Schur stability of $\check{A} - \check{L}\check{C}$, regardless of whether $\check{R}$ is positive definite or only positive semidefinite.
\end{Theorem}

\begin{proof}
From the LMI (\ref{eq:lmi_main}) with $\check{X} \succ 0$, applying the Schur complement with respect to the (1,1) block $\check{X}$ yields a $3 \times 3$ block inequality. Its (1,1) block gives:
\begin{equation}
(\check{X}\check{A} + \check{Y}\check{C})\check{X}^{-1}(\check{X}\check{A} + \check{Y}\check{C})^T \preceq \check{X} \label{eq:schur_extract}
\end{equation}
where the noise-related terms have been dropped since they are positive semidefinite. Pre- and post-multiplying (\ref{eq:schur_extract}) by $\check{X}^{-1}$ and substituting $\check{Y} = -\check{X}\check{L}$:
\begin{equation}
(\check{A} - \check{L}\check{C})\check{X}^{-1}(\check{A} - \check{L}\check{C})^T \preceq \check{X}^{-1} \label{eq:lyap_stability}
\end{equation}
This is a discrete Lyapunov inequality with $\check{X}^{-1} \succ 0$, which is a necessary and sufficient condition for $\check{A} - \check{L}\check{C}$ to be Schur stable. The condition depends solely on the closed-loop matrix $\check{A} - \check{L}\check{C}$ and the Lyapunov matrix $\check{X}^{-1}$, and does not involve $\check{R}$.
\end{proof}

\begin{Remark}[Roles of Observability and Semidefinite $\check{R}$]\label{rem:stability_roles}
The rank-deficiency of $\check{R}$ affects the filter \emph{design} computation: the standard DARE (\ref{eq:cyclic_dare}) requires inversion of $\check{C}\check{P}\check{C}^T + \check{R}$, which fails when $\check{R}$ is singular. The LMI formulation circumvents this entirely. Once a feasible LMI solution is obtained, the closed-loop stability follows from Theorem~\ref{thm:lmi_stability} independently of the rank of $\check{R}$. The periodic observability condition (Proposition~\ref{prop:observability}) is the system-theoretic condition that guarantees LMI feasibility with a non-degenerate solution. The constraint $\check{X} \succeq \epsilon I_{Nn}$ (LMI \#2) ensures that the Lyapunov matrix is bounded away from singularity, which is a numerical safeguard for solver robustness.
\end{Remark}

\textbf{Key Advantages of LMI Formulation:}
\begin{enumerate}
\item \textbf{Handles semidefinite $\check{R}$}: Works seamlessly with $\check{R} \succeq 0$
\item \textbf{Convex optimization}: Guarantees global optimality via semidefinite programming
\item \textbf{Multi-objective design}: Easy to incorporate pole placement and $l_2$-induced norm constraints (see Section~\ref{sec:multi_obj})
\item \textbf{Direct gain computation}: The Kalman gain is obtained as $\check{L} = -\check{X}^{-1} \check{Y}$
\item \textbf{Guaranteed stability}: Schur stability is ensured by the Lyapunov inequality structure (Theorem~\ref{thm:lmi_stability})
\end{enumerate}

\textbf{Note:} To validate the proposed LMI solution $ -\check{X}^{-1} \check{Y}$, I conducted numerical experiments with $\check{R} \succ 0$, where the standard DARE solution exists. Let $\check{K}_{\text{Ric}}$ denote the filter gain obtained from the DARE. The predictor-form gain is then $\check{L}_{\text{Ric}} = \check{A} \check{K}_{\text{Ric}}$. Numerical experiments confirm $\|-\check{X}^{-1} \check{Y} - \check{L}_{\text{Ric}}\|_F < 10^{-6}$. 

\begin{Remark}[Relationship to Classical Kalman Filter Optimality]\label{rem:kalman_optimality}
The proposed LMI formulation (\ref{eq:lmi_optimization}) minimizes an upper bound on $\mathrm{trace}(\check{P}_e)$ rather than $\mathrm{trace}(\check{P}_e)$ itself, since the Lyapunov inequality $\check{P}_e \preceq \check{X}^{-1}$ introduces conservatism. However, when the measurement noise covariance is positive definite ($\check{R} \succ 0$), the LMI solution coincides with the classical DARE-based Kalman filter gain, as confirmed both theoretically through the Kalman--LQR duality \cite{anderson1979optimal} and numerically (the above Note). When $\check{R}$ is only positive semidefinite, the standard DARE is not directly applicable because the inversion of $\check{C}\check{P}\check{C}^T + \check{R}$ in (\ref{eq:cyclic_gain}) fails. In this case, the LMI approach provides the best achievable trace upper bound under the Lyapunov inequality framework, representing a principled relaxation of the classical optimality criterion. The basic LMI design in this section can therefore be regarded as an optimal Kalman filter in the sense of upper-bound minimization, which reduces to the classical Kalman filter when the standard conditions are met.
\end{Remark}

\subsection{Weighted Filter Design}\label{sec:weighted}

The standard Kalman filter minimizes trace$(\check{X}^{-1})$, treating all state components equally. In practice, specific states may require higher estimation accuracy; for instance, in automotive navigation, position estimation is often more critical than acceleration estimation. This is achieved by introducing a weighting matrix $\check{\Gamma} \succ 0$ into the LMI constraint (\ref{eq:lmi_bound}):
\begin{equation}
\begin{bmatrix} \check{W} & \check{\Gamma}\\ \check{\Gamma} & \check{X} \end{bmatrix} \succeq 0 \label{eq:weighted_lmi}
\end{equation}
which ensures $\check{W} \succeq \check{\Gamma} \check{X}^{-1} \check{\Gamma}$ via Schur complement. Minimizing trace$(\check{W})$ then yields weighted optimal estimation, where the weight $\gamma_i^2$ is applied to each diagonal element $[\check{X}^{-1}]_{ii}$ for diagonal $\check{\Gamma} = \mathrm{diag}(\gamma_1, \ldots, \gamma_{Nn})$.

\section{Multi-objective Extensions}\label{sec:multi_obj}

\begin{Remark}[Naming Convention for Constrained Designs]\label{rem:luenberger}
The basic LMI design in Section~\ref{lmibased} minimizes trace$(\check{W})$, an upper bound on the estimation error covariance trace, and coincides with the classical Kalman filter when $\check{R} \succ 0$ (see Remark~\ref{rem:kalman_optimality}). The multi-objective extensions in this section retain the same trace minimization objective but restrict the feasible gain set through additional constraints (pole placement, $l_2$-induced norm bound). The resulting designs are therefore referred to as ``Kalman filter with \ldots\ constraint,'' reflecting that the trace minimization criterion of the Kalman filter is preserved while the admissible gain set is constrained. Note that these constrained solutions no longer achieve the unconstrained Kalman filter optimum; they are optimal within their respective constrained feasible sets. When the constraint becomes inactive (e.g., $\bar{r} \to 1$ for pole placement, or $\bar{\gamma} \to \infty$ for the $l_2$-induced norm bound), the solution reduces to the unconstrained optimal Kalman filter of Section~\ref{lmibased}.
\end{Remark}

A key advantage of the LMI framework is its ability to incorporate multiple design objectives simultaneously. This section presents two important extensions: pole placement constraints for convergence rate guarantees and Kalman filter design with $l_2$-induced norm constraints for balancing average and worst-case performance.

\subsubsection{Pole Placement Constraints}

To guarantee a minimum convergence rate for the estimation error, I constrain the eigenvalues of the error dynamics matrix $\check{A} - \check{L}\check{C}$ to lie within a specified region. For discrete-time systems, a common choice is the disk region:
\begin{equation}
\mathcal{D}_{\bar{r}} = \{z \in \mathbb{C} : |z| < \bar{r}\}, \quad 0 < \bar{r} < 1
\label{eq:disk_region}
\end{equation}
The parameter $\bar{r}$ determines the minimum decay rate: smaller $\bar{r}$ implies faster convergence but may require larger gains.

\begin{Theorem}[Pole Placement via LMI \cite{boyd1994lmi}]\label{poletheo}
The eigenvalues of $\check{A} - \check{L}\check{C}$ lie within $\mathcal{D}_{\bar{r}}$ if and only if there exists $\check{X} \succ 0$ such that:
\begin{equation}
(\check{A} - \check{L}\check{C})^T \check{X} (\check{A} - \check{L}\check{C}) - \bar{r}^2 \check{X} \prec 0
\label{eq:pole_lmi}
\end{equation}
\end{Theorem}

Using the variable transformation $\check{Y} = -\check{X}\check{L}$, the Schur complement yields the equivalent LMI:
\begin{equation}
\begin{bmatrix}
\bar{r}^2 \check{X} & \check{X}\check{A} + \check{Y}\check{C} \\
(\check{X}\check{A} + \check{Y}\check{C})^T & \check{X}
\end{bmatrix} \succ 0
\label{eq:pole_lmi_design}
\end{equation}

By Theorem~\ref{poletheo}, this constraint can be added to the optimization problem (\ref{eq:lmi_optimization}) to ensure the desired convergence rate.

\textbf{Design with Pole Constraints:}
\begin{equation}
\begin{aligned}
\minimize_{\check{X}, \check{Y}, \check{W}} \quad & \text{trace}(\check{W}) \\
\text{subject to} \quad & \text{LMI (\ref{eq:lmi_main}), (\ref{eq:lmi_stability}), (\ref{eq:lmi_bound})} \\
& \text{LMI (\ref{eq:pole_lmi_design}) for pole placement}
\end{aligned}
\label{eq:lmi_with_poles}
\end{equation}

\subsubsection{Kalman Filter with $l_2$-induced Norm Constraint}

The standard trace minimization corresponds to optimal Kalman filtering, which minimizes the expected estimation error under Gaussian noise. However, in practice, it is often desirable to also limit the worst-case error amplification. For time-invariant systems, this is characterized by the $H_\infty$ norm; for periodic systems, the appropriate generalization is the $l_2$-induced norm.

\begin{Remark}
For the cyclic system (\ref{eq:cyclic_A})--(\ref{eq:cyclic_C}), although the reformulated system is time-invariant, the underlying periodic structure suggests that the performance criterion is more precisely characterized as the $l_2$-induced norm rather than the $H_\infty$ norm. The LMI-based $l_2$-induced norm analysis and synthesis for periodically time-varying observers via cyclic reformulation have been established in \cite{okajima2019multirate,okajima2024design}.
\end{Remark}

Consider the error dynamics with performance output $z(k) = C_z e(k)$:
\begin{equation}
e(k+1) = (A - L_k S_k C) e(k) + Q^{1/2} d_w(k) - L_k S_k R^{1/2} d_v(k)
\end{equation}
where $d_w(k)$ and $d_v(k)$ are unit-covariance disturbances. Let $d(k) = [d_w(k)^T, d_v(k)^T]^T$ denote the combined disturbance. The $l_2$-induced norm from $d$ to $z$ is defined as:
\begin{equation}
\|G_{d \to z}\|_{l_2/l_2} = \sup_{\|d\|_2 \neq 0} \frac{\|z\|_2}{\|d\|_2}
\end{equation}

\begin{Theorem}[$l_2$-induced Norm via LMI \cite{okajima2019multirate}]\label{l2theo}
Using the variable transformation $\check{Y} = -\check{X}\check{L}$, the $l_2$-induced norm satisfies $\|G_{d \to z}\|_{l_2/l_2} < \gamma$ if there exists $\check{X} \succ 0$ such that:
\begin{equation}
\begin{bmatrix}
\check{X} & \check{X}\check{A} + \check{Y}\check{C} & \check{X}\check{Q}^{1/2} & \check{Y}\check{R}^{1/2} \\
(\check{X}\check{A} + \check{Y}\check{C})^T & \check{X} - \check{C}_z^T \check{C}_z & 0 & 0 \\
  (\check{Q}^{1/2})^T \check{X} & 0 & \gamma^2 I_{Nn} & 0\\
  (\check{R}^{1/2})^T \check{Y}^T & 0 & 0 & \gamma^2 I_{Nq}
\end{bmatrix} \succ 0
\label{eq:l2_lmi_design}
\end{equation}
where $\check{C}_z = \mathrm{diag}(C_z, \ldots, C_z)$ is the cyclic performance output matrix. For full state estimation error ($z = e$), set $\check{C}_z = I_{Nn}$.
\end{Theorem}

\begin{Remark}[Scaling for Feasibility]
The LMI (\ref{eq:l2_lmi_design}) admits a scaling degree of freedom: replacing $\check{C}_z$ with $\alpha \check{C}_z$ and $\gamma^2$ with $\alpha^2\gamma^2$ for any $\alpha > 0$ yields an equivalent $l_2$-induced norm condition. This scaling is useful when the (2,2) block $\check{X} - \check{C}_z^T \check{C}_z$ causes infeasibility; choosing $\alpha < 1$ relaxes the constraint $\check{X} \succ \check{C}_z^T \check{C}_z$, enabling a common Lyapunov matrix $\check{X}$ to satisfy both the Kalman filter LMI (\ref{eq:lmi_main}) and the $l_2$-induced norm constraint simultaneously.
\end{Remark}

\vspace{0.5em}
\noindent\textbf{Kalman Filter with $l_2$-induced Norm Constraint:}

Based on Theorem~\ref{l2theo}, the mixed design problem minimizes the estimation error covariance subject to an $l_2$-induced norm constraint:
\begin{equation}
\begin{aligned}
\minimize_{\check{X}, \check{Y}, \check{W}} \quad & \text{trace}(\check{W}) \\
\text{subject to} \quad & \text{LMI (\ref{eq:lmi_main}), (\ref{eq:lmi_stability}), (\ref{eq:lmi_bound})} \\
& \text{LMI (\ref{eq:l2_lmi_design}) with } \|G_{d \to z}\|_{l_2/l_2} < \bar{\gamma}
\end{aligned}
\label{eq:kf_with_l2_constraint}
\end{equation}
where $\bar{\gamma}$ is a prescribed upper bound on the $l_2$-induced norm.

\begin{Remark}
The (2,2) block $\check{X} - \check{C}_z^T \check{C}_z$ in (\ref{eq:l2_lmi_design}) arises from the bounded real lemma for the $l_2$-induced norm. When $\check{C}_z = I_{Nn}$, this becomes $\check{X} - I_{Nn}$, which requires $\check{X} \succ I_{Nn}$ for feasibility.
\end{Remark}

\subsubsection{Combined Multi-objective Design}

All the above extensions can be combined into a single optimization problem:
\begin{equation}
\begin{aligned}
\minimize_{\check{X}, \check{Y}, \check{W}} \quad & \text{trace}(\check{W}) \\
\text{subject to} \quad & \text{LMI (\ref{eq:lmi_main}): Kalman filter performance} \\
& \check{X} \succeq \epsilon I_{Nn} \quad \text{(positive definiteness)} \\
& \begin{bmatrix} \check{W} & I_{Nn} \\ I_{Nn} & \check{X} \end{bmatrix} \succeq 0 \quad \text{(covariance bound)} \\
& \text{LMI (\ref{eq:pole_lmi_design}): Pole placement in } \mathcal{D}_{\bar{r}} \\
& \text{LMI (\ref{eq:l2_lmi_design}): } \|G_{d \to z}\|_{l_2/l_2} < \bar{\gamma}
\end{aligned}
\label{eq:multi_objective}
\end{equation}

This formulation provides a flexible framework for designing multirate Kalman filters that satisfy multiple performance specifications simultaneously.

\begin{Remark}[Extension to Polytopic Uncertainty]\label{rem:polytopic}
The LMI-based framework has the potential for extension to systems with parametric uncertainty. When the system matrix $A$ belongs to a polytope defined by vertices $\{A^{(1)}, A^{(2)}, \ldots, A^{(M)}\}$, one can pursue a robust filter design by enforcing LMI constraints at each vertex with a common Lyapunov matrix $\check{X}$ and common gain variable $\check{Y}$. However, this extension requires reformulation of the error dynamics, since the nominal cancellation of the control input $u(k)$ in the estimation error equation no longer holds when the system matrix is uncertain. A detailed treatment is left for future work.
\end{Remark}

\subsection{Design Algorithm}

\textbf{Algorithm 1: LMI-based Multirate Kalman Filter Design}
\begin{enumerate}
\item \textbf{System Setup:} Construct cyclic matrices $\check{A}$, $\check{C}$, $\check{Q}$, $\check{R}$ from (\ref{eq:cyclic_A}), (\ref{eq:cyclic_C})

\item \textbf{Preprocessing:} Construct $\check{Q}^{1/2}$ with cyclic structure (place $Q^{1/2}$ in the same pattern as $\check{A}$), and compute $\check{R}^{1/2} = \mathrm{diag}(S_0 R^{1/2}, \ldots, S_{N-1} R^{1/2})$

\item \textbf{Specify Design Objectives:}
\begin{itemize}
\item Basic design: trace minimization only
\item With convergence rate: add pole constraint (\ref{eq:pole_lmi_design}) with specified $\bar{r}$
\item With robustness: add $l_2$-induced norm constraint (\ref{eq:l2_lmi_design}) with bound $\bar{\gamma}$
\item Multi-objective: use combined formulation (\ref{eq:multi_objective})
\end{itemize}

\item \textbf{Solve LMI Optimization:} Using SDP solver (e.g., MOSEK, SeDuMi, SDPT3):
\[
\begin{aligned}
\minimize_{\check{X}, \check{Y}, \check{W}} \quad & \text{trace}(\check{W}) \quad \text{(or multi-objective cost)}\\
\text{subject to} \quad & \text{Selected LMI constraints}
\end{aligned}
\]

\item \textbf{Recover Kalman Gain:}
\[
\check{L} = -\check{X}^{-1} \check{Y}
\]

\item \textbf{Extract Periodic Gains:} Obtain $L_k$, $k = 0, 1, \ldots, N-1$ from $\check{L}$ using cyclic indexing

\item \textbf{Verification:}
\begin{itemize}
\item Stability: $|\lambda_i(\check{A} - \check{L}\check{C})| < 1$ for all $i$
\item Pole placement (if specified): $|\lambda_i(\check{A} - \check{L}\check{C})| < \bar{r}$
\item $l_2$-induced norm bound (if specified): $\|G_{d \to z}\|_{l_2/l_2} < \bar{\gamma}$
\end{itemize}
\end{enumerate}

\textbf{Computational Complexity:}
\begin{itemize}
\item Basic LMI: $O((Nn)^{3.5})$ using interior-point methods
\item Multi-objective: Modest increase due to additional constraints
\item All computations performed offline; online filtering uses fixed periodic gains
\end{itemize}

\section{Numerical Example: Multirate Automotive State Estimation}\label{sec:numerical}

To demonstrate the effectiveness of the proposed multirate Kalman filter design method, I consider an automotive navigation system as a practical application. This example illustrates how the cyclic reformulation framework can be applied to a real-world multirate sensing scenario where GPS and wheel speed sensors operate at different sampling rates.

\subsection{System Configuration}

I consider an automotive navigation system with the following specifications:
\begin{itemize}
\item Sampling time: $\Delta t = 0.1$ s
\item Period: $N = 10$ (GPS at 1 Hz, wheel speed at 10 Hz)
\item State dimension: $n = 3$ (position, velocity, acceleration)
\item Measurement dimension: $q = 2$ (GPS position, wheel speed)
\end{itemize}

The vehicle dynamics are modeled by:
\begin{equation}
x(k) = \begin{bmatrix}
p(k)\\
v(k)\\
a(k)
\end{bmatrix}
\end{equation}
where $p(k)$, $v(k)$, and $a(k)$ represent position, velocity, and acceleration, respectively.

The system matrices are:
\begin{equation}
A = \begin{bmatrix}
1 & 0.1 & 0.005\\
0 & 1 & 0.1\\
0 & 0 & 0.8
\end{bmatrix}, \quad
B = \begin{bmatrix}
0\\
0\\
1
\end{bmatrix}
\end{equation}

\begin{equation}
C = \begin{bmatrix}
1 & 0 & 0\\
0 & 1 & 0
\end{bmatrix}
\end{equation}
where the first row corresponds to GPS position measurement and the second row corresponds to wheel speed velocity measurement.

The noise covariance matrices are:
\begin{equation}
Q = \mathrm{diag}(0.01, 0.1, 0.5), \quad R = \mathrm{diag}(1.0, 0.1)
\end{equation}
representing GPS accuracy of $\pm 1$ m and wheel speed accuracy of $\pm 0.316$ m/s.

The measurement selection matrix $S_k$ is:
\begin{equation}
S_k = \begin{cases}
\mathrm{diag}(1, 1) & \text{if } k \bmod 10 = 0 \text{ (GPS + wheel speed)}\\
\mathrm{diag}(0, 1) & \text{otherwise (wheel speed only)}
\end{cases}
\end{equation}

\subsection{Cyclic System Construction and Semidefinite $\check{R}$}

The cyclic measurement noise covariance $\check{R}$ is constructed according to (\ref{eq:cyclic_R}). The explicit computation yields:
\begin{align}
S_0 R S_0^T &= \begin{bmatrix} 1.0 & 0\\ 0 & 0.1 \end{bmatrix} \quad \text{(full rank)}\\
S_k R S_k^T &= \begin{bmatrix} 0 & 0\\ 0 & 0.1 \end{bmatrix}, \quad k = 1, \ldots, 9 \quad \text{(rank 1)}
\end{align}

\textbf{Verification of Semidefinite Nature:} The cyclic covariance $\check{R} \in \mathbb{R}^{20 \times 20}$ has:
\begin{equation}
\text{rank}(\check{R}) = 11 < 20 = Nq
\end{equation}
This confirms that $\check{R} \succeq 0$ but $\check{R} \not\succ 0$, validating the theoretical analysis in Section \ref{sec:design} and necessitating the LMI-based approach.

The observability matrix of the cyclic system has full rank:
\begin{equation}
\text{rank}(\mathcal{O}) = 30 = Nn
\end{equation}
with condition number $\text{cond}(\mathcal{O}) = 1.11 \times 10^1$, indicating well-conditioned observability.

\subsection{LMI-based Optimal Kalman Filter Design}

The LMI optimization problem (\ref{eq:lmi_optimization}) is solved using MATLAB's Robust Control Toolbox with the dual LQR formulation. The solver converges with:
\begin{equation}
\text{trace}(\check{W}) = 18.07 \quad \text{(upper bound on trace}(\check{X}^{-1})\text{)}
\end{equation}

The closed-loop eigenvalues satisfy:
\begin{equation}
\max_i |\lambda_i(\check{A} - \check{L}\check{C})| = 0.9673 < 1
\end{equation}
confirming stability of the estimation error dynamics. The periodic Kalman gains extracted from $\check{L}$ are shown in Table \ref{tab:gains}.

\begin{table}[t]
\caption{Periodic Kalman Gains (Optimal Kalman Filter)}
\label{tab:gains}
\centering
\begin{tabular}{|c|c|c|}
\hline
$k \bmod 10$ & Sensors Active & $L_k$ \\
\hline
0 & GPS + Wheel & $\begin{bmatrix}
0.2827 & 0.1017\\
0.0042 & 0.6979\\
0.0062 & 0.3755
\end{bmatrix}$ \\
\hline
1 & Wheel only & $\begin{bmatrix}
0 & 0.1094\\
0 & 0.6980\\
0 & 0.3757
\end{bmatrix}$ \\
\hline
5 & Wheel only & $\begin{bmatrix}
0 & 0.1148\\
0 & 0.6981\\
0 & 0.3758
\end{bmatrix}$ \\
\hline
\end{tabular}
\end{table}

When only wheel speed measurements are available ($k \bmod 10 \neq 0$), the first column of $L_k$ is zero, reflecting the absence of GPS measurements. The gains exhibit periodic variation, with larger values at GPS update times.

\subsection{Simulation Results}

\begin{figure}[!b]
\centering
\includegraphics[width=0.8\columnwidth]{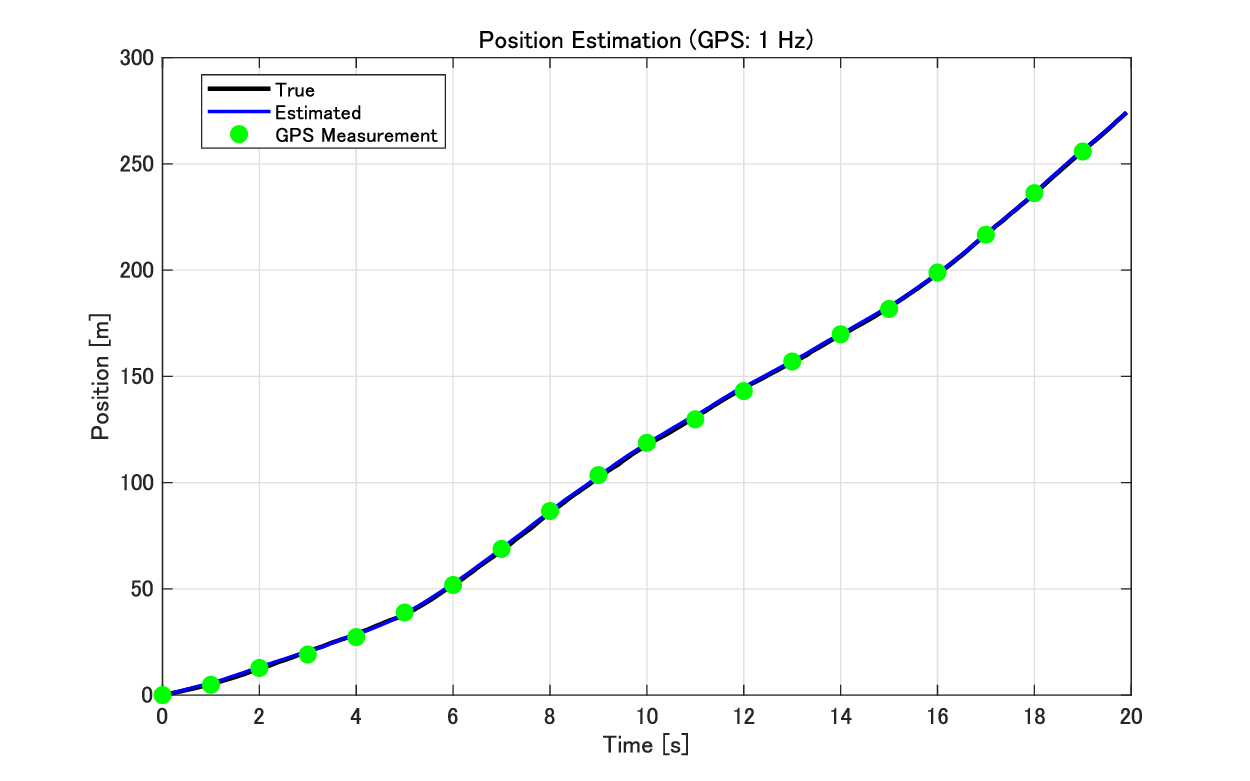}
\caption{Position estimation results. The green circles indicate GPS measurements available at 1 Hz.}
\label{fig:position}
\end{figure}
\begin{figure}[!b]
\centering
\includegraphics[width=0.8\columnwidth]{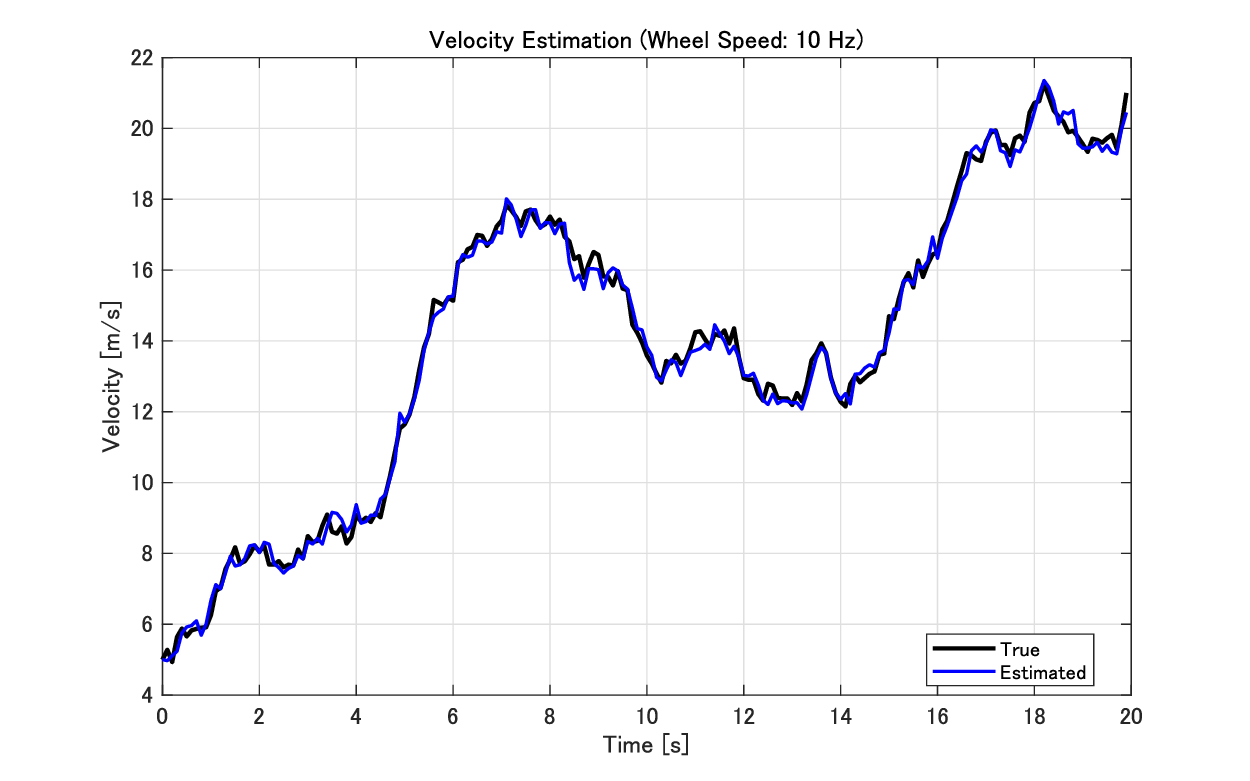}
\caption{Velocity estimation results. Wheel speed measurements are available at 10 Hz.}
\label{fig:velocity}
\end{figure}

\begin{figure}[!b]
\centering
\includegraphics[width=0.8\columnwidth]{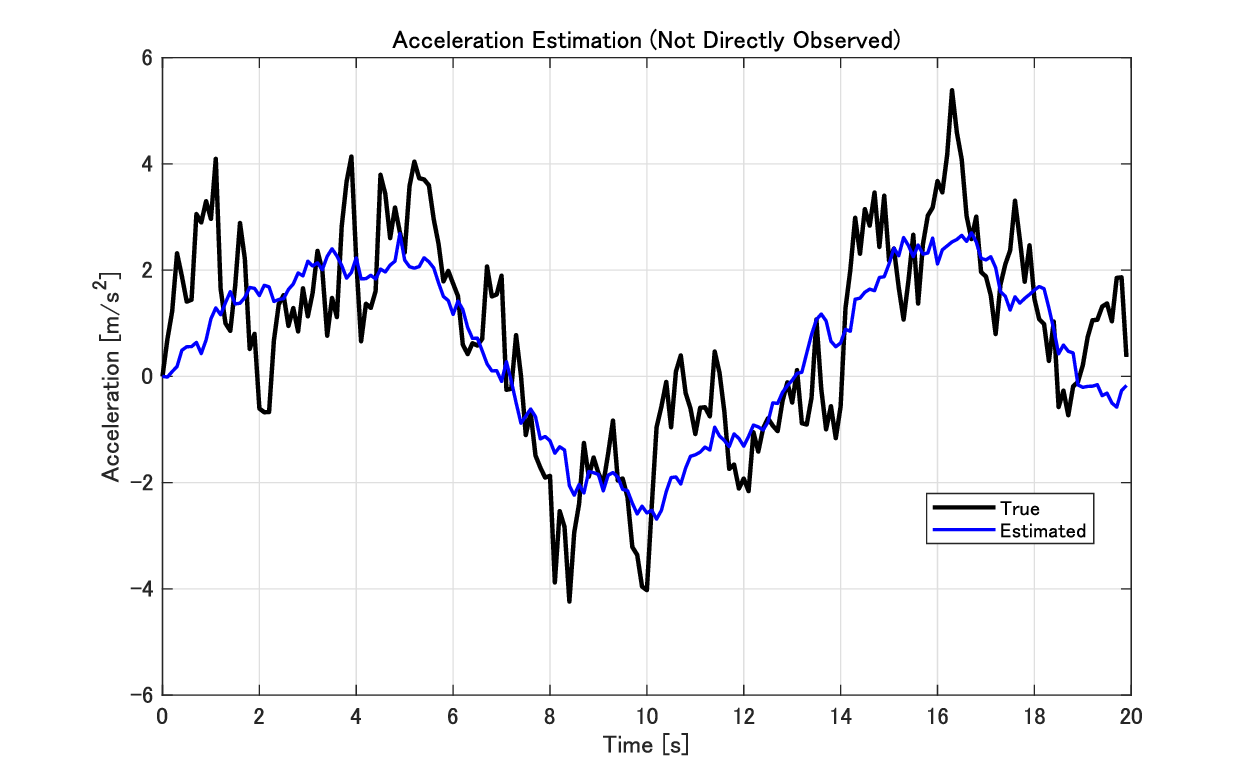}
\caption{Acceleration estimation results. Acceleration is not directly observed but estimated from the dynamics.}
\label{fig:acceleration}
\end{figure}

A simulation of 200 time steps (20 seconds) was conducted with initial state $x(0) = [0, 5, 0]^T$ (position 0~m, velocity 5~m/s), control input $u(k) = 0.5\sin(0.05k)$ (sinusoidal acceleration), and process and measurement noise drawn from zero-mean Gaussian distributions with covariances $Q$ and $R$.

Figs.~\ref{fig:position}--\ref{fig:acceleration} show a representative estimation trajectory for position, velocity, and acceleration, respectively. Position estimation errors are largest just before GPS updates (Fig.~\ref{fig:position}) and are corrected sharply when GPS measurements become available.

To statistically validate the filter performance, Monte Carlo simulations with 500 independent noise realizations were conducted. Each run uses the same system parameters and initial conditions, with independently generated process and measurement noise sequences. The RMSE for each run is computed over the steady-state interval (time steps 51--200), discarding the first 50 steps as transient. This cutoff is based on the maximum closed-loop eigenvalue magnitude of 0.967, corresponding to a time constant of approximately 30 steps.

\begin{table}[t]
\caption{Monte Carlo Simulation Results (500 runs, steady-state evaluation)}
\label{tab:monte_carlo}
\centering
\small
\begin{tabular}{|l|c|c|c|c|}
\hline
State & RMSE & 95\% Pctl & Theory UB & Ratio \\
\hline
Position & 0.561 m & [0.317, 0.919] & 0.587 m & 1.046 \\
Velocity & 0.431 m/s & [0.372, 0.490] & 0.432 m/s & 1.002 \\
Accel. & 1.125 m/s$^2$ & [0.889, 1.380] & 1.130 m/s$^2$ & 1.005 \\
\hline
\end{tabular}
\end{table}

Table~\ref{tab:monte_carlo} summarizes the Monte Carlo results. The column ``Theory UB'' lists the theoretical steady-state RMSE for each state, and the column ``Ratio'' is defined as Theory UB divided by the Monte Carlo mean RMSE. Because the cyclic error covariance upper bound $\check{X}^{-1} \in \mathbb{R}^{Nn \times Nn}$ contains $N$ diagonal blocks of size $n \times n$, one for each time step within the frame period, the theoretical RMSE for state $j$ is computed by averaging the corresponding diagonal entries over the $N$ time steps:
\begin{equation}
\text{Theory UB}_j = \sqrt{\frac{1}{N}\sum_{k=0}^{N-1}[\check{X}^{-1}]_{(kn+j,\, kn+j)}} \label{eq:theory_ub}
\end{equation}
This period-averaged value provides a fair comparison with the Monte Carlo mean RMSE, which is also computed over all time steps in the steady-state interval. Since the LMI formulation guarantees $\check{P}_e \preceq \check{X}^{-1}$ (see (\ref{eq:lmi_bound})), the theoretical values are guaranteed to be no less than the true steady-state RMSE. The Monte Carlo mean RMSE converges to this true value as the number of runs increases; therefore, the ratio exceeding unity for all three states confirms that the LMI solution provides a valid and consistent upper bound. The ratios are close to unity (1.002--1.046), indicating that the bound is tight with low conservatism.

Regarding estimation accuracy, the position RMSE (0.561~m) is well below the GPS noise standard deviation (1.0~m), confirming that multirate sensor fusion effectively reduces position estimation error. The velocity RMSE (0.431~m/s) exceeds the wheel speed sensor noise standard deviation (0.316~m/s); this is expected because the dominant error source for velocity estimation is not measurement noise but the propagation of acceleration process noise ($Q_{33} = 0.5$), which is consistent with the theoretical covariance analysis.

\subsection{Multi-objective Design Results}

To demonstrate the multi-objective design capabilities of the LMI framework, I present results for both Kalman filter with pole placement and Kalman filter with $l_2$-induced norm constraint.

\subsubsection{Kalman Filter with Pole Placement}

Using the pole placement constraint (\ref{eq:pole_lmi_design}), I investigate the trade-off between Kalman filter performance (trace minimization) and convergence rate (eigenvalue constraint $|\lambda| < \bar{r}$). Table \ref{tab:pole_tradeoff} and Fig.~\ref{fig:pole_tradeoff} summarize the results.

Compared to the optimal Kalman filter (trace$(\check{W}) = 18.07$), even a mild constraint $\bar{r} = 0.975$ increases trace$(\check{W})$ to 19.64. The trade-off is highly nonlinear: moderate constraints ($\bar{r} \approx 0.90$) roughly double the cost, while aggressive constraints ($\bar{r} < 0.80$) increase it by more than an order of magnitude.

\subsubsection{Kalman Filter with $l_2$-induced Norm Constraint}

Using the $l_2$-induced norm constraint (\ref{eq:l2_lmi_design}), I investigate the trade-off between average performance (trace minimization) and worst-case robustness ($l_2$-induced norm). $C_z = \sqrt{0.1}I_{n}$ is selected. The minimum achievable $l_2$-induced norm is $\gamma_{\text{opt}} = 1.0214$.

Table \ref{tab:l2_tradeoff} and Fig.~\ref{fig:l2_tradeoff} show the trade-off results.  $\gamma_{\text{opt}}$ is the result of $l_2$-induced norm optimization \cite{okajima2019multirate}. As $\bar{\gamma}$ approaches $\gamma_{\text{opt}}$, trace$(\check{W})$ increases from 18.71 (at $\bar{\gamma}/\gamma_{\text{opt}} = 10$) to 34.65 (at $\bar{\gamma}/\gamma_{\text{opt}} = 1.01$), while the maximum eigenvalue magnitude decreases from 0.961 to 0.815, indicating improved stability margins at the cost of estimation performance.

\subsubsection{Comparison of Optimal Gains}

Table \ref{tab:gain_comparison} compares the periodic Kalman gains $L_0$ (at GPS + wheel speed update) for the optimal Kalman filter and the optimal $l_2$-induced norm design~\cite{okajima2023multirate}. The $l_2$-optimal design yields a position gain of $1.000$ (compared to $0.283$ for the Kalman filter), indicating that worst-case robustness optimization prioritizes GPS measurements more aggressively, while the velocity and acceleration gains also increase to achieve a smaller maximum eigenvalue magnitude ($0.653$ vs.\ $0.967$).

\begin{table}[!t]
\caption{Kalman Filter with Pole Placement: Trade-off}
\label{tab:pole_tradeoff}
\centering
\begin{tabular}{|c|c|c|}
\hline
$\bar{r}$ & trace$(\check{W})$ & max$|\lambda|$ \\
\hline
0.975 & 19.64 & 0.953 \\
0.950 & 24.91 & 0.910 \\
0.925 & 31.73 & 0.868 \\
0.900 & 41.19 & 0.826 \\
0.875 & 55.10 & 0.785 \\
0.850 & 76.45 & 0.746 \\
0.825 & 110.4 & 0.707 \\
0.800 & 165.9 & 0.707 \\
0.775 & 259.4 & 0.705 \\
0.750 & 422.1 & 0.684 \\
\hline
\end{tabular}
\end{table}

\begin{figure}[!t]
\centering
\includegraphics[width=0.8\columnwidth]{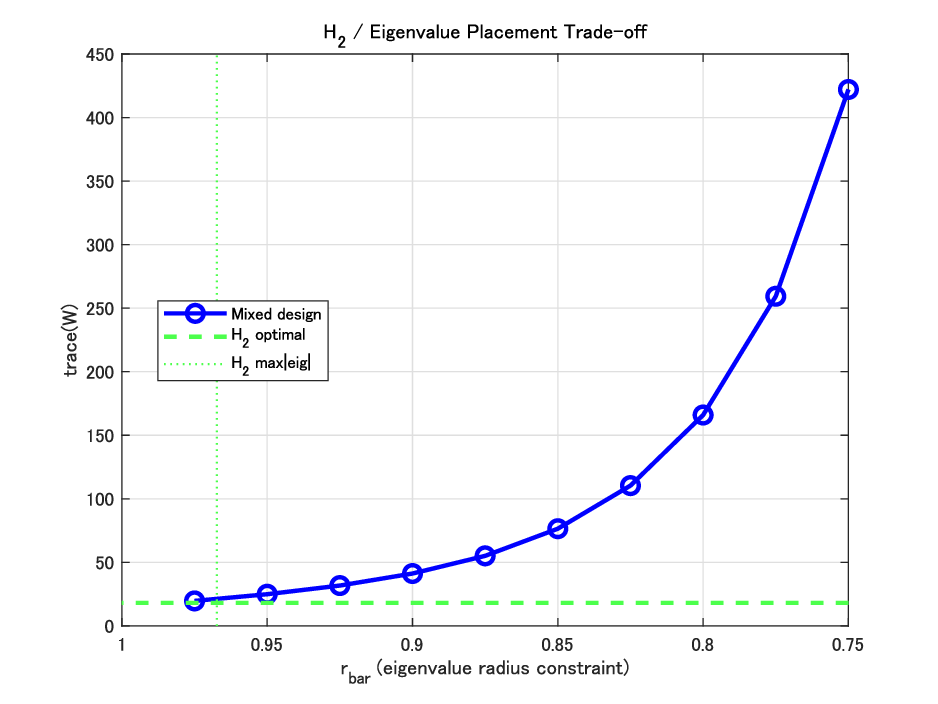}
\caption{Kalman filter with pole placement: trace$(\check{W})$ versus pole constraint $\bar{r}$.}
\label{fig:pole_tradeoff}
\end{figure}

\begin{table}[!t]
\caption{Kalman Filter with $l_2$-induced Norm Constraint: Trade-off}
\label{tab:l2_tradeoff}
\centering
\begin{tabular}{|c|c|c|c|}
\hline
$\bar{\gamma}/\gamma_{\text{opt}}$ & $\bar{\gamma}$ & trace$(\check{W})$ & max$|\lambda|$ \\
\hline
10.0 & 10.21 & 18.71 & 0.961 \\
5.0 & 5.11 & 19.12 & 0.957 \\
3.0 & 3.06 & 19.99 & 0.949 \\
2.0 & 2.04 & 21.40 & 0.936 \\
1.5 & 1.53 & 23.08 & 0.921 \\
1.3 & 1.33 & 24.25 & 0.911 \\
1.2 & 1.23 & 25.04 & 0.905 \\
1.1 & 1.12 & 26.10 & 0.897 \\
1.05 & 1.07 & 27.63 & 0.892 \\
1.01 & 1.03 & 34.65 & 0.815 \\
\hline
\end{tabular}
\end{table}

\begin{figure}[!t]
\centering
\includegraphics[width=0.8\columnwidth]{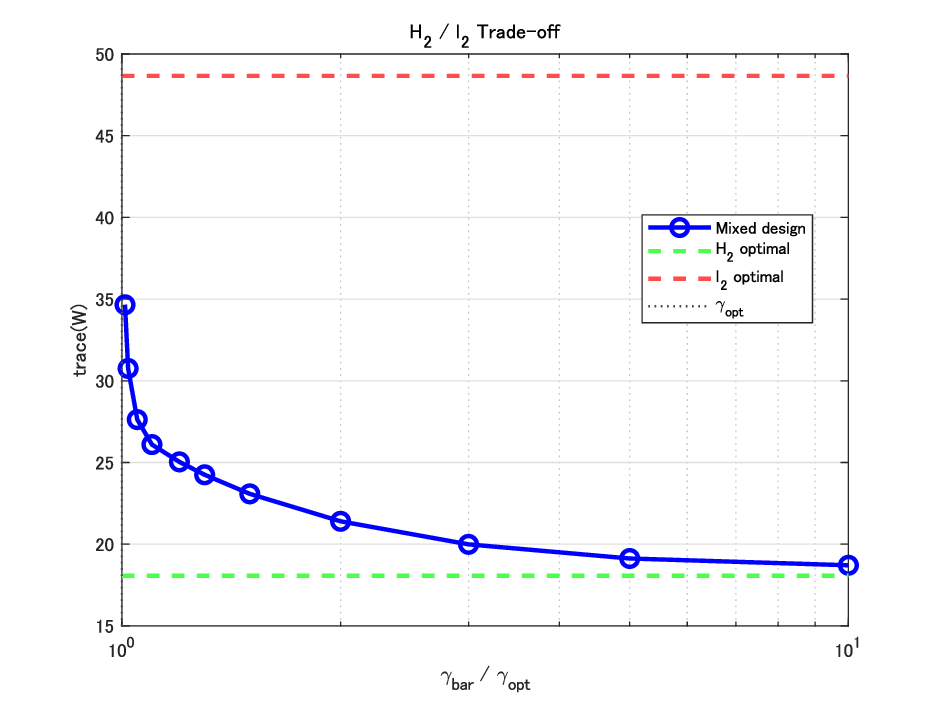}
\caption{Kalman filter with $l_2$-induced norm constraint: trace$(\check{W})$ and maximum eigenvalue magnitude vs. $\bar{\gamma}/\gamma_{\text{opt}}$.}
\label{fig:l2_tradeoff}
\end{figure}

\subsection{Design Guidelines}

Based on the numerical results, the following design strategy is recommended. For standard applications, the optimal Kalman filter design (Section~\ref{lmibased}) provides the best average estimation performance under stochastic disturbances characterized by the noise covariances $Q$ and $R$. When fast convergence is required, adding a pole placement constraint with $\bar{r} = 0.7$--$0.8$ guarantees a minimum decay rate at the cost of a moderate increase in the estimation error covariance. When limiting worst-case performance degradation is critical, the $l_2$-induced norm constraint with $\bar{\gamma} = 1.5$--$2.0 \times \gamma_{\text{opt}}$ bounds the worst-case gain from disturbances $(\check{d}_w, \check{d}_v)$ to the estimation error. When multiple requirements must be satisfied simultaneously, the combined formulation (\ref{eq:multi_objective}) provides a systematic means of balancing all criteria within a single convex optimization. For robustness against uncertain system matrices, the polytopic uncertainty extension discussed in Remark~\ref{rem:polytopic} is the appropriate approach.

The LMI-based approach provides a unified framework for incorporating these diverse requirements within a single convex optimization problem, enabling systematic design of multirate Kalman filters tailored to specific application needs. The MATLAB and Python implementation of the proposed LMI-based multirate Kalman filter design is available at \cite{okajima2025github}.

\begin{table}[t]
\caption{Comparison of $L_0$ for Different Design Criteria}
\label{tab:gain_comparison}
\centering
\small
\begin{tabular}{|c|c|c|}
\hline
Design & $L_0$ & max$|\lambda|$ \\
\hline
Optimal Kalman filter & $\begin{bmatrix}
0.283 & 0.102\\
0.004 & 0.698\\
0.006 & 0.376
\end{bmatrix}$ & 0.967 \\
\hline
Optimal $l_2$-induced norm \cite{okajima2023multirate} & $\begin{bmatrix}
1.000 & 0.086\\
0.002 & 0.787\\
0.006 & 0.558
\end{bmatrix}$ & 0.653 \\
\hline
\end{tabular}
\end{table}

\section{Conclusion}\label{sec:conclusion}

This paper has presented an LMI-based design framework for multirate steady-state Kalman filters using cyclic reformulation. In such systems, the Kalman gains converge to periodic steady-state values that repeat every frame period. The central finding is that the cyclic measurement noise covariance $\check{R}$ becomes positive semidefinite---not positive definite---when sensors operate at different rates, preventing direct application of standard DARE-based methods. The proposed dual LQR formulation with LMI optimization naturally handles this semidefinite structure while enabling optimal estimation error covariance minimization and multi-objective extensions including pole placement and $l_2$-induced norm constrained design.

Numerical validation using state estimation in an automotive navigation system with GPS (1~Hz) and wheel speed (10~Hz) sensors, based on Monte Carlo simulation with 500 independent noise realizations, confirmed that the proposed filter achieves a mean position RMSE of 0.561~m, well below the GPS noise standard deviation of 1.0~m (ratio 0.56). The theoretical steady-state RMSE values, obtained by period-averaging the diagonal elements of the $Nn$-dimensional error covariance upper bound $\check{X}^{-1}$ from the LMI solution, exceed the Monte Carlo averages for all states (ratio 1.002--1.046), confirming both the validity and the tightness of the LMI upper bound. The trade-off analysis demonstrated that the LMI framework provides systematic tuning between convergence rate and estimation accuracy.

The framework is applicable in principle to general linear time-invariant systems with deterministic, periodic measurement patterns, where the sensor availability schedule is known a priori and repeats every frame period. The numerical example in this paper serves as a proof-of-concept demonstration with a 3-state system; validation with higher-dimensional systems is an important direction for future work. Other promising directions include robust filter design under polytopic uncertainty (see Remark~\ref{rem:polytopic} in Section~\ref{sec:multi_obj}), generalization to periodically repeating non-uniform sampling intervals, extensions to descriptor systems \cite{nikoukhah1992descriptor} and systems with finite delay, sensor configuration switching strategy design guided by the LMI performance metrics, and extensions to nonlinear systems via extended Kalman filtering.

\section*{Acknowledgment}
The author acknowledges the use of Claude (Anthropic) for improving 
the readability and language clarity of the manuscript, and for 
technical verification of mathematical expressions. Specifically, 
Claude was used to refine the English expression in the Introduction and Literature Review, and to check dimensional 
consistency and numerical accuracy of equations. All technical content, 
research methodology, results, and conclusions are entirely the 
author's own work.

\end{document}